\documentclass[aps,twocolumn,superscriptaddress,nofootinbib]{revtex4-1}

\usepackage{enumitem}
\usepackage{amsmath}
\usepackage{amssymb}
\usepackage{amsthm}
\usepackage{bm}
\usepackage{graphicx}
\usepackage{hyperref}
\usepackage{url}
\usepackage{xcolor}

\newcommand{\tr}[0]{\textrm{tr}}
\newcommand{\bra}[1]{\langle #1 |}
\newcommand{\ket}[1]{| #1 \rangle}
\newcommand{\braket}[2]{\langle #1 | #2 \rangle}

\theoremstyle{definition}
\newtheorem{observation}{Observation}
\newtheorem{theorem}{Theorem}
\newtheorem{proposition}{Proposition}
\newtheorem{lemma}{Lemma}
\newtheorem{conjecture}{Conjecture}

\begin{document}

\title{Fundamental limits on concentrating and preserving tensorized quantum resources}

\author{Jaehak Lee}
\email{jaehaklee@kias.re.kr}
\affiliation{School of Computational Sciences, Korea Institute for Advanced Study, Seoul 02455, Korea}
\author{Kyunghyun Baek}
\affiliation{School of Computational Sciences, Korea Institute for Advanced Study, Seoul 02455, Korea}
\affiliation{Eletronics and Telecommunications Research Institute, Daejeon 34129, Korea}
\author{Jiyong Park}
\affiliation{School of Basic Sciences, Hanbat National University, Daejeon 34158, Korea}
\author{Jaewan Kim}
\affiliation{School of Computational Sciences, Korea Institute for Advanced Study, Seoul 02455, Korea}
\author{Hyunchul Nha}
\email{hyunchul.nha@qatar.tamu.edu}
\affiliation{Department of Physics, Texas A \& M University at Qatar, P.O. Box 23874, Doha, Qatar}

\begin{abstract}
Quantum technology offers great advantages in many applications by exploiting quantum resources like nonclassicality, coherence, and entanglement. In practice, an environmental noise unavoidably affects a quantum system and it is thus an important issue to protect quantum resources from noise. In this work, we investigate the manipulation of quantum resources possessing the so-called tensorization property and identify the fundamental limitations on concentrating and preserving those quantum resources. We show that if a resource measure satisfies the tensorization property as well as the monotonicity, it is impossible to concentrate multiple noisy copies into a single better resource by free operations. Furthermore, we show that quantum resources cannot be better protected from channel noises by employing correlated input states on joint channels if the channel output resource exhibits the tensorization property. We address several practical resource measures where our theorems apply and manifest their physical meanings in quantum resource manipulation.
\end{abstract}

\maketitle

\section{\label{sec:introduction}Introduction}

% quantum resource theory
Quantum technology has rapidly grown providing a variety of applications such as quantum computation \cite{Nielson}, quantum cryptography \cite{Gisin2007,Lo2014}, and quantum metrology \cite{Giovanetti2011}. Tasks in quantum information processing take advantages over their classical counterparts by exploiting quantum \emph{resources}. For instance, quantum teleportation makes use of pre-shared entanglement between two distant parties \cite{Pirandola2015}. Recently, nonclassicality has also been introduced as a resource enhancing the metrological power \cite{Yadin2018,Kwon2019,Ge2020,Tan2019}. It is generally a crucial issue to identify quantum resources essential to obtain quantum advantages. Quantum resource theories (QRTs) \cite{Chitambar2019} provide a comprehensive framework for quantifying and manipulating quantum resources of interest. QRTs have been developed for a diverse range of quantum resources including entanglement \cite{Horodecki09}, coherence \cite{Baumgratz2014,Winter2016,Streltsov2017}, nonclassicality \cite{Yadin2018,Kwon2019,Ge2020,Tan2019,Tan2017}, quantum non-Gaussianity \cite{Takagi2018,Albarelli2018,Park2019}, {\it etc}.. In each QRT, one defines free operations which do not create resource under consideration. Then, one of the fundamental problems in QRT is to investigate whether one resource state can be converted into another by using free operations only.

% noisy environment
In a realistic situation, quantum resources are contaminated with noise due to interaction with environments. To overcome practical noise, several approaches have been proposed to protect resources from channel noise and to distill useful resources from noisy copies. For example, entanglement distillation protocol aims to obtain maximally entangled states from $N$ copies of non-maximally entangled states using local operations and classical communications (LOCC) only \cite{Bennett1996}. Determining conversion rates between resource states has recently been studied as a central problem in QRTs \cite{Winter2016,Brandao2015,Regula2018,Fang2018,Lami2019,Liu2019}.

% tensorization property
A quantum resource may be quantified by a proper resource measure $ R $, whose properties lead to some fundamental rules in manipulating quantum resources. One of the important properties is the monotonicity which implies that resources are nonincreasing under free operations. In this work, along with the monotonicity, we focus on the so-called \emph{tensorization property}, namely $ R \left( \rho \otimes \sigma \right) = \max \left\{ R(\rho), R(\sigma) \right\} $. Tensorization property naturally arises in several studies on quantifying measures, not restricted to the quantum resource theory. For instance, it has served as a key tool to find several applications in classical information theory, e.g., non-interactive distribution simulation problem \cite{Kamath2012}, distributed channel coding problem \cite{Wang2011}, and also in quantum information theory, e.g., the distillation of nonlocal correlation by wiring \cite{Beigi2015}. By incorporating the monotonicity and the tensorization property, we explore fundamental limitations in manipulating quantum resources.

% resource concentration
First, we show the limitation on the concentration of quantum resources. While one typically aims at maximal resource state in the distillation protocol, we are here interested in a scenario where one aims to produce more resourceful output by consuming noisy copies, referred to as \emph{resource concentration}. This is meaningful as one can obtain quantum advantages in quantum informational tasks not necessarily employing the maximally resourceful states. Furthermore, in some resource theories like those on nonclassicality and non-Gaussianity, maximally resourceful states are not well defined due to the infinite dimension of quantum systems. Further, while the distillation rate is a central problem in the resource distillation, we rather focus on the possibility of the resource concentration. The key idea of our investigation is the tensorization property of resources, which naturally leads to the limitation on resource concentration. The limitation states that a number of noisy copies cannot be concentrated into a better single copy using free operations under a resource measure satisfying both monotonicity and tensorization properties. The extent of limitation on resource concentration may depend on whether the monotonicity holds for deterministic or probabilistic free operations.

% resource preservability
Next, we consider a measure that quantifies the resource of channel output to study the behavior of quantum resources under channel noise. We investigate the tensorization property of channel output resource to show the limitation on the resource preservability. If the channel-output resource measure satisfies the tensorization property, we find that employing correlated input state is not advantageous in protecting quantum resources from noise. We prove that the nonclassicality depth satisfies such a property and find a strong evidence for the tensorization property of maximal coherence, which is an important coherence measure in the problem of coherence distillation under strictly incoherent operations \cite{Lami2019}.

This paper is organized as follows. 
In Sec. \ref{sec:QRT}, we briefly introduce a general framework of QRTs and address the properties of quantum resource measures. In Sec. \ref{sec:concentration}, we show the limitation on resource concentration under the tensorization property and the monotonicity. As examples, we present several resource measures where our theorem applies, and especially we investigate two measures of nonclassicality, i.e. the nonclassicality depth \cite{Lee1991,Sabapathy2016} and the metrological power of nonclassicality \cite{Yadin2018,Kwon2019,Ge2020}. In Sec. \ref{sec:preservability}, we investigate the tensorization property of channel output resources and apply it to the nonclassicality depth and the maximal coherence. We finally conclude with some remarks in Sec. \ref{sec:discussion}.

\section{\label{sec:QRT}Quantum resource theory}

% free state
% free operation
%% maximal set
%% physically meaningful set
Let us begin by briefly introducing the basic formalism of QRTs. A resource theory is individually determined by the restrictions imposed on freely available operations, which defines the set of free operations $ \mathcal{O}$. The states that can be accessible using the free operations in $ \mathcal{O}$ are considered as free states, forming a convex set $ \mathcal{F} $. All the states other than free states are considered as resource states. The golden rule of QRTs is that free operations never create resources. That is, if $ \rho \in \mathcal{F} $ and $ \Phi \in \mathcal{O} $, then $ \Phi(\rho) \in \mathcal{F} $. One might consider the maximal set of free operations $ \tilde{\mathcal{O}} $, the collection of all operations satisfying this golden rule. On the other hand, we may consider a subset $ \mathcal{O} $ that does not necessarily contain all possible free operations, but includes only those free operations with operational motivation. 

Let us take the resource theory of entanglement as an example. Here entangled states are considered as resources, while $ \mathcal{F} $ consists of separable states. The maximal set $ \tilde{\mathcal{O}} $ must be constructed by all non-entangling operations \cite{Harrow2003,Brandao2008,Chitambar2020}, while LOCC is a physically motivated subset $ \mathcal{O} $ of non-entangling operations widely used in the study of entanglement manipulation. Choosing a different set of free operations exhibits a different property of quantum resource. For example, manipulation of entanglement is irreversible under LOCC, while it is reversible under all non-entangling operations \cite{Brandao2008,Berta2022, footnote1}.
 
% resource monotones
%% monotonicity
%% strong monotonicity
%% selective monotonicity
% 
% other properties
%% faithfulness
%% convexity
%% additivity
%% tensorization property!
%%% compare with additivity
Resource can be quantified by a proper resource measure. In the following, we list the desired properties of a resource measure $ R $.
	\begin{itemize}
	\item[(\textit{C}1a)] Vanishing for free states: $ R(\rho) \geqslant 0 $ with equality if $ \rho \in \mathcal{F} $.
	\item[(\textit{C}1b)] Faithfulness: $ R(\rho) = 0 $ if and only if $ \rho \in \mathcal{F} $.
	\item[(\textit{C}2a)] Monotonicity: $ R $ does not increase under trace-preserving free operations, i.e., $ R(\rho) \geqslant R\left( \Phi(\rho) \right) $.
	\item[(\textit{C}2b)] Strong monotonicity: $ R $ does not increase on average under conditional free operations, i.e., $ R(\rho) \geqslant \sum_i p_i R(\sigma_i) $, where $ p_i = \tr \left[ \Phi_i (\rho) \right] $ and $ \sigma_i = \Phi_i (\rho) / p_i $.
	\item[(\textit{C}2c)] Monotonicity under postselection of free operations: $ R $ does not increase even under trace non-preserving free operations, i.e., $ R(\rho) \geqslant R\left( \Phi_i(\rho) / p_i \right) $.
	\item[(\textit{C}3)] Convexity: $ R (\sum_i p_i \rho_i) \leqslant \sum_i p_i R(\rho_i) $.
	\end{itemize}

The property (\textit{C}1a) is essential in order to avoid fake detection of resources. Monotonicity (\textit{C}2a) is a fundamental property of resource measure implying that free operations cannot create or increase resource. Strong monotonicity (\textit{C}2b) is a stronger requirement when a map can be decomposed into sum of free operations as $ \Phi (\cdot) = \sum_i \Phi_i (\cdot) $ and one has access to the outcome $ i $ individually. It is straightforward to show that when (\textit{C}2b) is combined with convexity (\textit{C}3), (\textit{C}2a) is automatically satisfied. The property (\textit{C}2c) is an even stronger requirement implying that resource measure is nonincreasing under each probabilistic free operation. We primarily require the two properties (\textit{C}1a) and (\textit{C}2a) to be satisfied. Then we may further require a stricter restriction, (\textit{C}2b) or (\textit{C}2c), depending on the situation as we shall address below.

We also introduce an important property considered to derive our main result.
	\begin{itemize}
	\item[(\textit{C}4)] Tensorization property: $ R \left( \rho \otimes \sigma \right) = \max \{ R(\rho), R(\sigma) \} $
	\end{itemize}
This property is quite different from the other frequently considered property, i.e. additivity $ R \left( \rho \otimes \sigma \right) = R(\rho) + R(\sigma) $, which is usually obeyed in many resource measures. While additive measures quantify the total amount of resources, measures satisfying the tensorization property tends to estimate the strongest ability in certain tasks. For example, the metrological power of nonclassicality \cite{Yadin2018,Kwon2019,Ge2020} estimates the optimal performance of sensing displacement in a single direction among all possible directions involving multi-mode fields in phase space.

\section{\label{sec:concentration}No-go theorem for resource concentration}

% illustration of resource concentration
The goal of resource concentration is to obtain an output state with a higher degree of resource from noisy resources. Suppose one has noisy resource states $ \rho_1, \rho_2, \cdots, \rho_N $, which hardly provides advantage in a quantum information task that exploits the resource quantified by $ R $. If one is able to concentrate resources into an output copy $ \sigma $ such that $ R(\sigma) > R(\rho_j) $ for any $ j $, the output state can perform the task better than any of $ \rho_j $.

In the resource concentration, one may aim at producing a single output state $ \sigma = \sum_i p_i \sigma_i = \sum_i \Phi_i ( \otimes_{j=1}^N \rho_j ) $ using free operations $ \Phi (\cdot) = \sum_i \Phi_i (\cdot) $. Let us denote the dimensions of input states $ \rho_j $ and output state $ \sigma$ as $ d_1, d_2, \cdots, d_N $ and $ d_\textrm{out} $, respectively. When we consider resources in infinite dimensions, we set all $ d_i $'s and $ d_\textrm{out} $ to be infinite. On the other hand, when we deal with resources in finite dimensions, the dimensions are not necessarily the same as long as free operations and the resource measure are well defined for different dimensions. The map $ \Phi $ transforms a $N$-partite state into a single-party state. This procedure may include measuring or discarding subsystems other than a desired output subsystem, which is allowed in free operations.

One can say that concentration is successful if
	\begin{equation} \label{eq:concentration}
	R (\sigma) > \max_j R (\rho_j) .
	\end{equation}
If a resource measure $ R $ satisfies both monotonicity (\textit{C}2a) and tensorization property (\textit{C}4), one can readily show that achieving Eq. (\ref{eq:concentration}) is impossible. Furthermore, if $ R $ also possesses a stronger monotonicity, (\textit{C}2b) or (\textit{C}2c), we find a stricter no-go theorem for resource concentration.
	\begin{observation} \label{obs:concentration}
	(No-go theorem for resource concentration)
	\begin{enumerate}[label=(\alph*)]
	\item If $ R $ satisfies (\textit{C}2a) and (\textit{C}4), then resource cannot be concentrated by the deterministic free operation, i.e.,
		\begin{equation}
		R (\sigma) \leqslant \max_j R (\rho_j) .
		\end{equation}
	\item If $ R $ satisfies (\textit{C}2b) and (\textit{C}4), then resource cannot be concentrated on average, i.e.,
		\begin{equation}
		\sum_i p_i R (\sigma_i) \leqslant \max_j R (\rho_j) .
		\end{equation}
	\item If $ R $ satisfies (\textit{C}2c) and (\textit{C}4), then resource cannot be concentrated by a probabilistic free operation, i.e.,
		\begin{equation}
		R (\sigma_i) \leqslant \max_j R (\rho_j) , ~~~ \forall i .
		\end{equation}
	\end{enumerate}
	\end{observation}
We here remark that the above no-go theorems hold regardless of the number of input copies $ n $. For example, when one tries to obtain a target state with resource $ R(\sigma_\textrm{T}) $ by a probabilistic concentration, Observation \ref{obs:concentration}(b) implies that the success probability is bounded by $ P_\textrm{succ} \leqslant \max_j R (\rho_j) / R(\sigma_\textrm{T}) $. In this case, increasing the number of input states does not increase the success probability. In the following, we introduce resource measures to which Observation \ref{obs:concentration} applies and discuss their physical meanings in quantum information tasks.

\subsection{\label{sec:ND}Nonclassicality depth}

In continuous-variable systems, an $n$-mode state $ \rho $ is called classical if it can be written as a convex mixture of coherent states, that is,
	\begin{equation} \label{eq:classical}
	\rho = \int d^{2n}\bm{\alpha} P(\bm{\alpha}) \ket{\bm{\alpha}}\bra{\bm{\alpha}} , ~~~ P(\bm{\alpha}) \geqslant 0 .
	\end{equation}
Here $ \bm{\alpha} = ( \textrm{Re}[\alpha_1], \textrm{Im}[\alpha_1], \cdots, \textrm{Re}[\alpha_n], \textrm{Im}[\alpha_n] )^T $ is a $2n$-dimensional vector in phase space and $ \ket{\bm{\alpha}} = \hat{D}(\bm{\alpha}) \ket{0} $ represents $n$-mode coherent state, where $ \hat{D}(\bm{\alpha}) = \bigotimes_{i=1}^n \exp \left( \alpha_i \hat{a}^\dagger_i + \alpha_i^\ast \hat{a}_i \right) $ is the multimode displacement operator and $ \hat{a}_i $($ \hat{a}^\dagger_i $) the annihilation (creation) operator associated with $i$th mode.  In the resource theory of nonclassicality, the set of free states includes all classical states (\ref{eq:classical}) having positive Glauber-Sudarshan {\it P} functions \cite{Tan2019,Glauber1963,Sudarshan1963}.

We adopt the following operations as free operations. (1) passive linear unitaries and displacements, (2) addition of classical ancilla modes, (3) classical measurement (projection onto coherent states), (4) classical mixing. Passive linear unitaries are photon-number-conserving unitary operations implemented with beam splitters and phase shifters. Any classical measurements can be realized by projection onto coherent states followed by coarse-graining of measurement outcomes. Our set of free operations does not cover all classicality-preserving operations \cite{Gehrke2012}, however, it covers physically motivated operations that are deemed easy to implement. Note that slightly different choices of free operations are considered in the resource theory of nonclassicality \cite{Yadin2018, Kwon2019, Ge2020} as well as in the Gaussian work extraction problem \cite{Singh2019}.

The nonclassicality depth was first introduced by C. T. Lee for a single mode \cite{Lee1991}, generalized also to multimode cases \cite{Sabapathy2016}. An $n$-mode state $ \rho $ can be represented by the so-called $s$-parametrized quasiprobability function \cite{Barnett2003}
	\begin{equation} \label{eq:spara}
	W_\rho (\bm{\alpha}; s) = \left( \frac{2}{\pi(1-s)} \right)^n \int d^{2n}\bm{\beta} P_\rho (\bm{\beta}) e^{-\frac{2}{1-s} \left| \bm{\beta}-\bm{\alpha} \right|^2 } ,
	\end{equation}
where $ P_\rho (\bm{\beta}) $ is the Glauber-Sudarshan {\it P} function of the state $ \rho $. When $ s=-1 $, $ W_\rho (\bm{\alpha}; s) $ corresponds to the Husimi Q function, which is always non-negative. Substituting $ s = 1-2\tau $, we may rewrite Eq. (\ref{eq:spara}) as
	\begin{equation}
	W_\rho (\bm{\alpha}; \tau) = \left( \frac{1}{\pi\tau} \right)^n \int d^{2n}\bm{\beta} P_\rho (\bm{\beta}) e^{-\frac{1}{\tau} \left| \bm{\beta}-\bm{\alpha} \right|^2 } ,
	\end{equation}
Then the nonclassicality depth is defined by
	\begin{equation}
	\tau_m (\rho) = \inf \left\{ \tau | W_\rho (\bm{\alpha}; \tau) \textrm{ is positive} \right\} .
	\end{equation}
Operationally, $ \tau_m $ can be interpreted as the minimum amount of additive thermal noise which makes $ \rho $ classical.

Despite its early origin, the properties of nonclassicality depth has not been thoroughly studied yet. We here show that the nonclassicality depth satisfies several desired properties listed in Sec. \ref{sec:QRT} (See Appendix \ref{sec:NDproperty} for proof).
	\begin{proposition}
	The nonclassicality depth $ \tau_m $ satisfies the properties (\textit{C}1a,b) faithfulness, (\textit{C}2a-c) monotonicity under both deterministic and conditional free operations, and (\textit{C}4) tensorization property.
	\end{proposition}
As the nonclassicality depth satisfies the monotonicity (\textit{C}2a-c) and tensorization property (\textit{C}4), Observation \ref{obs:concentration} implies the following no-go theorem.
	\begin{theorem}
	The nonclassicality depth $ \tau_m $ cannot be concentrated by classical operations even probabilistically.
	\end{theorem}

Let us discuss some practical examples. For Gaussian states, the nonclassicality depth is given by
	\begin{equation}
	\tau_m (\rho_\textrm{G}) = \max \left\{ 0, \tfrac{1}{2} - \lambda_\textrm{min}(\bm{V}) \right\} ,
	\end{equation}
where $ \lambda_\textrm{min}(\bm{V}) $ denotes the minimum eigenvalue of the covariance matrix $ \bm{V} $ of $ \rho_\textrm{G} $, or the minimum variance over all different quadratures. Note that the nonclassicality depth is only determined by $ \lambda_\textrm{min}(\bm{V}) $ representing the degree of squeezing. Our theorem reproduces the no-go theorem for Gaussian squeezing distillation developed in \cite{Yadin2018,Kraus2003,Lami2018}, while different measures for nonclassicality are employed there.

The nonclassicality depth becomes maximal for any pure non-Gaussian state, that is, $ \tau_m (\ket{\psi_\textrm{NG}}\bra{\psi_\textrm{NG}}) = 1 $ \cite{Lutkenhaus1995}. For a lossy single-photon state $ \rho_\textrm{loss} = (1-q)\ket{0}\bra{0} + q\ket{1}\bra{1} $, the nonclassicality depth becomes $ \tau_m (\rho_\textrm{loss}) = q $. Our no-go theorem not only shows that it is impossible to distill a pure nonclassical state from lossy single-photon states, but also that the single-photon fraction cannot be increased by classical operations. The latter result agrees with the one given in \cite{Berry2010}. We can generalize this result to Fock states $ \ket{n} $ under a loss channel with transmissivity $ t $. Because the nonclassicality depth decreases monotonically as $ t $ decreases, it is impossible to obtain a pure Fock state from lossy Fock states.

\subsection{\label{sec:MP}Metrological power of nonclassicality}

Recently, nonclassicality has been studied as a resource quantifying the quantum advantage in metrological tasks \cite{Yadin2018,Kwon2019,Ge2020}. The phase space for $n$-mode bosonic system is described by $2n$ quadratures $ \hat{\bm{R}} = ( \hat{x}_1, \hat{p}_1, \cdots, \hat{x}_n, \hat{p}_n )^T $. When the state is displaced along the direction $ \bm{\mu} = ( \textrm{Re}[\mu_1], \textrm{Im}[\mu_1], \cdots, \textrm{Re}[\mu_n], \textrm{Im}[\mu_n] )^T $, which is a real $2n$-dimensional unit vector, the variance of the estimator is lower bounded by the quantum Cram\'{e}r-Rao bound \cite{Kwon2019,Braunstein1994}
	\begin{equation}
	(\Delta \theta)^2 \geqslant \frac{1}{4\bm{\mu}^T \bm{F} \bm{\mu}} .
	\end{equation}
$ \bm{F} $ is the quantum Fischer information (QFI) matrix of which elements are given by
	\begin{equation}
	\bm{F}_{kl} = \frac{1}{2} \sum_{i,j} \frac{(\lambda_i-\lambda_j)^2}{\lambda_i+\lambda_j} \bra{i} \hat{\bm{R}}_k \ket{j} \bra{j} \hat{\bm{R}}_l \ket{i} ,
	\end{equation}
where $ \lambda_i $ and  $ \ket{i} $ are eigenvalues and eigenstates of $ \rho $, respectively. For a later convenience, we use QFI scaled by a factor $ \frac{1}{4} $ compared to the usual definition, as adopted in \cite{Yadin2018,Ge2020}. Then the metrological power of nonclassicality is defined as \cite{Yadin2018,Kwon2019,footnote2}
	\begin{equation}
	F_1 (\rho) = \max \left\{ \lambda_\textrm{max} (\bm{F}) - \tfrac{1}{2} , 0 \right\} ,
	\end{equation}
where $ \lambda_\textrm{max} (\bm{F}) $ is the maximum eigenvalue of $ \bm{F} $. Operationally, $ F_1 $ quantifies the maximal advantage in displacement estimation along a certain direction among all possible directions in phase space.

For pure states, the metrological power of nonclassicality satisfies (\textit{C}1b), that is, $ F_1 = 0 $ only for coherent states. However, for mixed states, it satisfies only (\textit{C}1a) but not (\textit{C}1b) in general. Nevertheless, $ F_1 $ is a useful quantifier of nonclassicality due to its operational interpretation as well as the monotonic property. The monotonicity (\textit{C}2a) is proved in \cite{Yadin2018} and \cite{Kwon2019}, respectively, with a slightly different set of free operations. In \cite{Yadin2018}, it is shown that $ F_1 $ satisfies (\textit{C}1a) under the set of free operations defined similarly to what we employed in the previous section. On the other hand, Ref. \cite{Yadin2018} employs a much broader set of measurements, that is all destructive measurements where measured systems are discarded. The tensorization property (\textit{C}4) is proved in \cite{Kwon2019} as well. Now we can apply the Observation \ref{obs:concentration}(a) to obtain the following theorem.
	\begin{theorem}
	The metrological power of nonclassicality $ F_1 $ cannot be concentrated under deterministic classical operations.
	\end{theorem}

A probabilistic concentration of $ F_1 $ is possible, but the strong monotonicity (\textit{C}2b) of $ F_1 $ is not proved yet. Instead another form of monotonicity is proved in \cite{Yadin2018}, which states that $ p_i F_1 (\sigma_i) \leqslant F_1 (\otimes_{j=1}^N \rho_j) $ under classical operations and destructive measurements without feedforward. In this form, we only consider the output state of single measurement outcome without taking the average over all outcomes as done in (\textit{C}2b). Then, using the tensorization property (\textit{C}4), we formulate the limitation on probabilistic concentration as $ p_i F_1 (\sigma_i) \leqslant \max_j F_1 (\rho_j) $. When one tries to obtain the target state $ \sigma_\textrm{T} $ by probabilistic concentration, the success probability is bounded by the following theorem.
	\begin{theorem}
	In probabilistic concentration of $ F_1 $, the success probability is upper bounded as
		\begin{equation} \label{eq:pbound}
		P_\textrm{succ} \leqslant \frac{\max_j F_1 (\rho_j)}{F_1 (\sigma_\textrm{T})} .
		\end{equation}
	\end{theorem}
As noted in Sec. \ref{sec:concentration}, the upper bound does not depend on the number of input copies.

Let us recall the nonclassicality concentration of cat states $ \ket{\psi_c(\alpha)} = \frac{1}{\sqrt{\mathcal{N}_\alpha}} ( \ket{\alpha} + \ket{-\alpha} ) $, where $ \mathcal{N}_\alpha $ is the normalization factor, discussed in \cite{Yadin2018}. One aims to obtain a larger cat state $ \ket{\psi_c(\sqrt{2}\alpha)} $ from a pair of cat states $ \ket{\psi_c(\alpha)}^{\otimes2} $ assuming $ \alpha \gg 1 $. Performing 50:50 beam splitter interaction followed by vacuum projection on one mode, one obtains $ \ket{\psi_c(\sqrt{2}\alpha)} $ at the other mode with probability $ 1/2 $. Since the metrological power of the cat state is given by $ F_1 (\ket{\psi_c(\alpha)}) = 2|\alpha|^2 $, this scheme saturates the inequality $ P_\textrm{succ} \leqslant F_1(\ket{\psi_c(\alpha)}) / F_1(\ket{\psi_c(\sqrt{2}\alpha)}) $. Interestingly, while a single pair of input states is considered in this scheme, increasing the number of input states will never increase the success probability. There still remains the possibility that the cat state is probabilistically amplified from a single copy, but such protocols have not been reported yet.

One might wonder if the no-go theorems for two different measures of nonclassicality contradict each other. The contradiction does not occur, however, because the nonclassicality depth and the metrological power look into different aspects of nonclassicality. For example, the nonclassicality depth of a cat state is always $ 1 $ regardless of the amplitude $ \alpha $ while the metrological power increases monotonically with $ \alpha $. Therefore, the probabilistic amplification of cat state does not violate the no-go theorem for concentration of the nonclassicality depth. It is also possible to concentrate the metrological power from a pair of lossy single-photon state $ \rho_\textrm{loss}^{\otimes 2} $. Performing beamsplitter interaction followed by vacuum projection, one obtains the output state $ \sigma_\textrm{out} = \frac{1}{P} \left( (1-q)^2\ket{0}\bra{0} + q(1-q)\ket{1}\bra{1} + \frac{q^2}{2}\ket{2}\bra{2} \right) $ with the probability $ P = 1-q+\frac{q^2}{2} $.
	\begin{figure}[t]
	\centering \includegraphics[clip=true, width=\columnwidth]{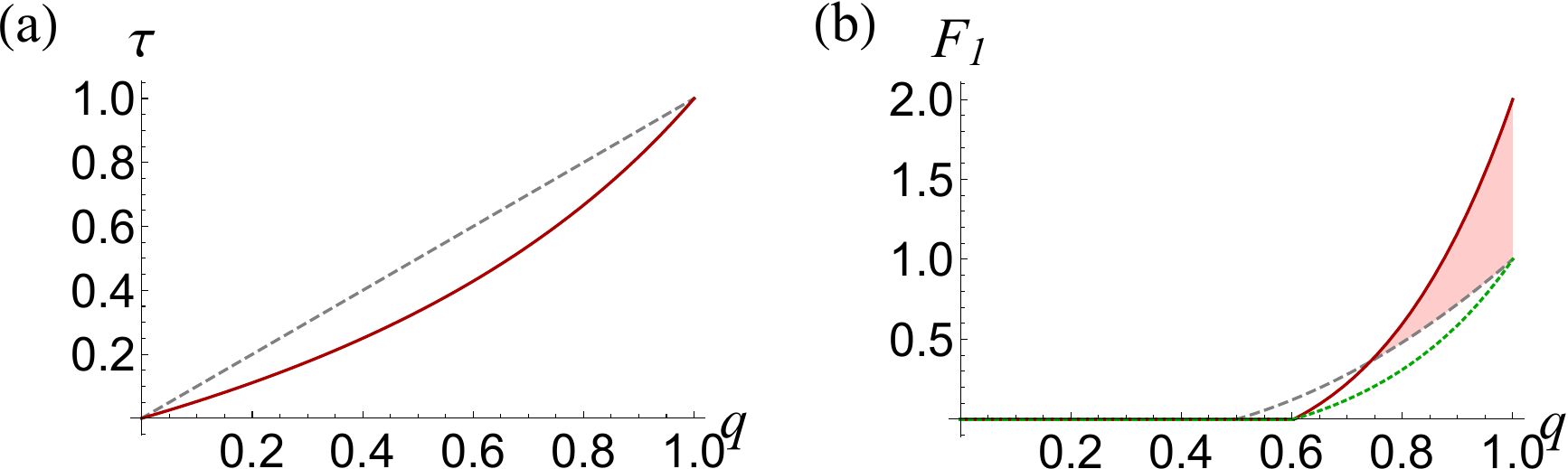}
	\caption{\label{fig:lossy} Plot illustrating (a) the nonclassicality depth and (b) the metrolgical power. Dashed curves correspond to input lossy single-photon state $ \rho_\textrm{loss} $. Red solid curved correspond to output state $ \sigma_\textrm{out} $ after vacuum projection on the other mode. Shaded region in (b) represents successful concentration. Green dotted curve represents $ P F_1 (\sigma_\textrm{out}) $.}
	\end{figure}
We show that one can attain the successful concentration $ F_1 (\sigma_\textrm{out}) > F_1 (\rho_\textrm{loss}) $ when $ q \gtrsim 0.7419 $ as shown in Fig. \ref{fig:lossy}(b), while the nonclassicality depth cannot be concentrated as shown in Fig. \ref{fig:lossy}(a). In Fig. \ref{fig:lossy}(b), by comparing $ F_1 (\sigma_\textrm{out}) $ and $ P F_1 (\rho_\textrm{loss}) $, it is observed that the bound (\ref{eq:pbound}) is strictly obeyed and that it is saturated only when $ q = 1 $.

It must also be noted that measurements considered as free operations are different, destructive measurements in this section and only classical measurements in Sec. \ref{sec:ND}. Consider that a two-mode squeezed vacuum is generated by mixing position- and momentum-squeezed states by 50:50 beamsplitter. One may try to concentrate nonclassicality into one mode by performing measurement on the other mode. Classical measurements cannot accomplish this task due to the no-go theorem for the nonclassicality depth. On the other hand, highly nonclassical Fock state can be attained by photon counting while the success probability is bounded by (\ref{eq:pbound}).

\subsection{Gaussian quantum resource theory}

The no-go theorem for Gaussian resource distillation was developed recently \cite{Lami2018} in a similar context to our work. In the Gaussian regime, a state is fully characterized by its displacement $ \bm{s} $ and covariance matrix $ \bm{V} $. In Gaussian QRT, the set of free states is defined as $ \mathcal{F}_\textrm{G} = \mathcal{F} \cap \mathcal{G} $, where $ \mathcal{G} $ denotes the set of Gaussian states. For a given covariance matrix $ \bm{V} $, a resource measure is defined as
	\begin{equation}
	\kappa_\mathcal{F} (\bm{V}) = \min \left\{ t \geqslant 1 ~ | ~ t\bm{V} \in \mathcal{V}_\mathcal{F} \right\} ,
	\end{equation}
where $ \mathcal{V}_\mathcal{F} $ denotes the set of covariance matrices corresponding to free states. For example, $ \kappa_\mathcal{F} $ becomes a Gaussian entanglement measure when $ \mathcal{F} $ is the set of separable states, or a Gaussian nonclassicality measure when $ \mathcal{F} $ is the set of classical states. It was shown that $ \kappa_\mathcal{F} $ satisfies the monotonicity (\textit{C}2a) and the tensorization property (\textit{C}4). Interestingly, because any Gaussian measurements on Gaussian states can be described in a deterministic way \cite{Eisert2002,Fiurasek2002,Giedke2002}, (\textit{C}2a) also implies (\textit{C}2c). Therefore Observation \ref{obs:concentration}(c) holds for $ \kappa_\mathcal{F} $. This result reproduces the no-go theorem for Gaussian squeezing \cite{Kraus2003} and the no-go theorem for Gaussian entanglement distillation \cite{Eisert2002,Fiurasek2002,Giedke2002}.

\subsection{\label{sec:maxcoh}Maximal coherence}

In discrete-variable systems, QRT of coherence has been extensively investigated \cite{Streltsov2017}. If one considers the maximal set of incoherent operations $ \tilde{\mathcal{O}} $, so called maximally incoherent operations (MIO), conversion of coherence is reversible in the asymptotic limit \cite{Brandao2015,Berta2022}. However, if one considers a smaller set of incoherent operations, we can find the irreversible behavior of coherence. The set of strictly incoherent operations (SIO) is a widely employed set of free operations \cite{Yadin2016,Chitambar2016} due to their simple structure. An operation is called SIO if each Kraus operator $ K_i $ and its adjoint $ K_i^\dagger $ are both incoherent. Lami \textit{et al.} introduced a measure called the maximal coherence \cite{Lami2019}, defined as
	\begin{equation}
	\eta (\rho) = \max_{i \neq j} \frac{\left| \rho_{ij} \right|}{\sqrt{\rho_{ii}\rho_{jj}}} ,
	\end{equation}
which satisfies the monotonicity (\textit{C}2c) under SIO. Also it obeys the tensorization property (\textit{C}4) so we obtain Observation \ref{obs:concentration}(c).

Using the property of maximal coherence, Ref. \cite{Lami2019} derived important theorems on coherence distillation under SIO. In the coherence distillation, one aims to convert a number of copies $ \rho $ into an $m$-dimensional maximally coherent state $ \Psi_m $ under a set of free operations $ \mathcal{O} $. The fidelity of distillation characterizes the error in the distillation as
	\begin{equation}
	F_\mathcal{O} (\rho, m) = \sup_{\Lambda \in \mathcal{O}} F \left( \Lambda (\rho) , \Psi_m \right) ,
	\end{equation}
where $ F ( \rho, \sigma) = \left( \textrm{tr} \sqrt{\sqrt{\rho} \sigma \sqrt{\rho}} \right)^2 $. The distillable coherence is the maximal rate of obtaining maximally coherent qubit state with vanishing error, that is,
	\begin{equation}
	C_{d,\mathcal{O}} (\rho) = \sup \left\{ r \middle| \lim_{n\to\infty} F_\mathcal{O} \left( \rho^{\otimes n} , 2^{rn} \right) = 1 \right\} .
	\end{equation}
We here restate the theorems on SIO coherence distillation.
	\begin{theorem} \label{thm:cohdistill}
	\begin{enumerate}[label=(\alph*)]
	\item \cite[Theorem 1]{Lami2019} A state $ \rho $ is distillable under SIO, that is, $ C_{d,\textrm{SIO}} (\rho) > 0 $, if and only if $ \eta (\rho) = 1 $ in Eq. (15).
	\item \cite[Theorem 3]{Lami2019} The fidelity of asymptotic distillation is given by
		\begin{equation} \label{eq:fdistill}
		\lim_{n \to \infty} F_\textrm{SIO} ( \rho^{\otimes n}, 2 ) = \frac{1 + \eta(\rho)}{2} .
		\end{equation}
	\end{enumerate}
	\end{theorem}
The condition $ \eta (\rho) = 1 $ means that there exists a submatrix of the density matrix $ \rho $ that corresponds to a pure coherent state. Note that Theorem \ref{thm:cohdistill} is fully described with the quantity $ \eta $. This result will be recalled in Sec. \ref{sec:outMC} where we study the output coherence of channels.

\section{\label{sec:preservability}Resource preservability of quantum channels}

% noisy channel problem
In the previous section, we have shown that if we have noisy resource states in a product state [Output state in Fig. \ref{fig:prep} (a)], their resource cannot be concentrated by free operations. Then, one may wonder whether it is possible to preserve more resources against noise by initially preparing correlated resource states as illustrated in Fig. \ref{fig:prep} (b).
	\begin{figure}[t]
	\centering \includegraphics[clip=true, width=0.9\columnwidth]{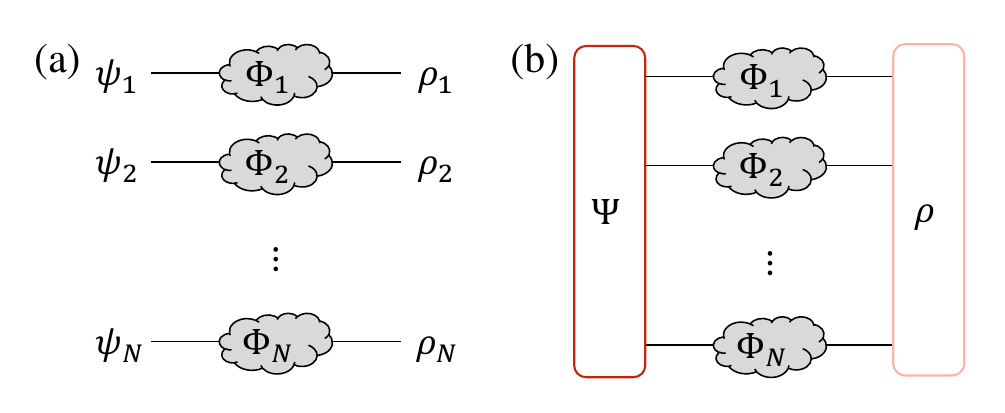}
	\caption{\label{fig:prep} (a) Resources are prepared in a product state and then undergo independent channel noises. (b) Resources are prepared in a correlated state before channel noises.}
	\end{figure}
To address this problem, we introduce a measure which estimates the output resource of channels, defined as
	\begin{equation}
	\tilde{R} (\Phi) = \max_\rho \left\{ R \left( \Phi (\rho) \right) \right\} ,
	\end{equation}
where the maximization is over all possible input states.

% comparison to common channel resource theory
We remark on a difference between our measure $ \tilde{R} $ and other measures in channel resource theories. In channel resource theories \cite{Chitambar2019,Dana2017,Takagi2019,Liu2020}, one is interested in how much resource channels can create. For example, the resource-generating power \cite{Dana2017,Mani2015,Zhuang2018} is one of the widely studied channel resource measures. In this approach, free operations $ \Phi \in \mathcal{O} $ are considered as free channel resources with $ R (\Phi) = 0 $. However, when studying the effect of noisy channels, this approach is not appropriate, because a noisy environment generally does not  create resources. In contrast, we are interested in how much resource channels can preserve. For example, the identity channel perfectly preserves resources, while it is a resource-nongenerating channel belonging to $ \mathcal{O} $. Similar arguments have been made in recent studies of channel resource theories \cite{Hsieh2020,Saxena2020}. Particularly, Ref. \cite{Hsieh2020} investigated the axiomatic properties of resource preservability measures. In this sense, our measure $ \tilde{R} $ is appropriate to study resources preserved by noisy channels.

% tensorization property
Now, by investigating the tensorization property, we have the following observation.
	\begin{observation} \label{obs:preserve}
	If a channel measure $ \tilde{R} $ satisfies the tensorization property, that is,
		\begin{equation} \label{eq:chantensor}
		\tilde{R} \left( \Phi_1 \otimes \Phi_2 \right) = \max \left\{ \tilde{R} (\Phi_1) , \tilde{R} (\Phi_2) \right\} ,
		\end{equation}
	employing correlated input states does not preserve more resources in the output states than employing product input states.
	\end{observation}
It is straightforward to show that once the state resource measure $R$ has the property $ R ( \rho_1 \otimes \rho_2 ) > \max \left\{ R (\rho_1) , R (\rho_2) \right\} $, Eq. (\ref{eq:chantensor}) for the channel resource $\tilde{R}$ does not hold. The tensorization property of state resource measure (\textit{C}4) is a necessary condition for the tensorization property of channel output resource measure (\ref{eq:chantensor}). However, the converse is not straightforward because the optimization of $ \tilde{R} \left( \Phi_1 \otimes \Phi_2 \right) $ is taken in a larger Hilbert space with correlated input states allowed than the optimization of $ \tilde{R} (\Phi_1) $ and of $ \tilde{R} (\Phi_2) $. We investigate output resource measures to which Observation \ref{obs:preserve} may apply and show their physical meaning in resource preparation.

\subsection{\label{sec:NDchannel}Output nonclassicality depth of channels}

Let us consider the output nonclassicality depth of channels defined as
	\begin{equation}
	\tilde{\tau}_m (\Phi) = \max_\rho \left\{ \tau_m \left( \Phi (\rho) \right) \right\} .
	\end{equation}
Note that a similar measure was employed in the framework of process output nonclassicality \cite{Sabapathy2016}. In what follows, we prove the tensorization property of $ \tilde{\tau}_m $, which implies that multiple use of channels with correlated or entangled input is not helpful to preserve the nonclassicality. We first show the result in the extreme case where $ \tilde{\tau}_m (\Phi) = 0 $.

\subsubsection{Nonclassicality-breaking channels}

Because $ \tau_m $ is a faithful measure, $ \tilde{\tau}_m (\Phi) = 0 $ implies that the channel $ \Phi $ outputs only classical states. That is, the channel is nonclassicality-breaking. The characterization of nonclassicality-breaking channels (NBCs) have been studied for the class of bosonic Gaussian channels \cite{Ivan2013,Sabapathy2015}. In the following, we investigate general properties of NBC, not constrained to bosonic Gaussian channels. Note that the properties of NBC are intrinsic properties of the channel, which does not rely on the nonclassicality measure.

To investigate the properties of NBC, we use the following lemma (See Appendix \ref{sec:NBCproof} for proof).
	\begin{lemma} \label{lem:NBCexp}
	Any nonclassicality-breaking channel $ \Phi_\textrm{NB} $ can be expressed in the form
		\begin{equation} \label{eq:NBCexp}
		\Phi_\textrm{NB} (\rho) = \int d^{2n}\bm{\alpha} \tr \left[ \rho M_{\bm{\alpha}} \right] \ket{\bm{\alpha}}\bra{\bm{\alpha}} ,
		\end{equation}
	where $ n $ is the number of output modes and $ \{ M_{\bm{\alpha}} \}_{\bm{\alpha}} $ a set of POVM operators with $\int d^{2n}\bm{\alpha}M_{\bm{\alpha}}=I$.
	\end{lemma}
Using Lemma \ref{lem:NBCexp}, a multiple use of NBCs, e.g. $ \Phi_1 $ and $ \Phi_2 $, can be expressed as a simple extension of Eq. (\ref{eq:NBCexp}) as
	\begin{eqnarray}
	\Phi_1 \otimes \Phi_2 (\rho) & = & \int d^{2n_1}\bm{\alpha_1} d^{2n_2}\bm{\alpha_2} \tr \left[ \rho ( M_{\bm{\alpha_1}} \otimes M_{\bm{\alpha_2}} ) \right] \nonumber \\
	& & \qquad \times \ket{\bm{\alpha_1}}\bra{\bm{\alpha_1}} \otimes \ket{\bm{\alpha_2}}\bra{\bm{\alpha_2}} .
	\end{eqnarray}
We then have the following proposition.
	\begin{proposition} \label{prop:NBCtensor}
	If both $ \Phi_1 $ and $ \Phi_2 $ are nonclassicality-breaking, then $ \Phi_1 \otimes \Phi_2 $ is nonclassicality-breaking as well.
	\end{proposition}

\subsubsection{Tensorization property}

Let us introduce a single-mode additive thermal noise channel $ \mathcal{E}_\delta $, which can be understood as random displacement with a Gaussian distribution in phase space. When $ \mathcal{E}_\delta^{\otimes n} $ acts on an $n$-mode state $ \rho $, it yields the output state as
	\begin{equation}
	\mathcal{E}_\delta^{\otimes n} (\rho) = \frac{1}{(\pi\delta)^n} \int d^{2n}\bm{\gamma} e^{-\frac{1}{\delta}|\bm{\gamma}|^2} \hat{D}(\bm{\gamma}) \rho \hat{D}^\dagger(\bm{\gamma}) .
	\end{equation}
The $P$ function is transformed under $ \mathcal{E}_\delta $ accordingly as
	\begin{eqnarray} \label{eq:noiseP}
	P_{\mathcal{E}_\delta^{\otimes n} (\rho)} (\bm{\alpha}) & = & \frac{1}{(\pi\delta)^n} \int d^{2n}\bm{\gamma} e^{-\frac{1}{\delta}|\bm{\gamma}|^2} P_\rho (\bm{\alpha}-\bm{\gamma}) \nonumber \\
	& = & \frac{1}{(\pi\delta)^n} \int d^{2n}\bm{\gamma'} e^{-\frac{1}{\delta}|\bm{\gamma}'-\bm{\alpha}|^2} P_\rho (\bm{\gamma}') \nonumber \\
	& = & W_\rho (\bm{\alpha}; \delta) .
	\end{eqnarray}
From Eq. (\ref{eq:noiseP}) in conjunction with the definition of nonclassicality depth, we readily see that those states with $ \tau_m (\rho) \leqslant \delta $ become classical after the channel noise $ \mathcal{E}_\delta^{\otimes n} $ while states remain nonclassical if $ \tau_m > \delta $. We can make a similar conclusion for the output nonclassicality depth of channels, that is, $ \mathcal{E}_\delta^{\otimes n} \circ \Phi $ is nonclassicality-breaking if and only if $ \tilde{\tau}_m (\Phi) \leqslant \delta $.

Now we obtain the following theorem.
	\begin{theorem}
	The output nonclassicality depth of channels $ \tilde{\tau}_m $ satisfies the tensorization property, that is,
		\begin{equation}
		\tilde{\tau}_m ( \Phi_1 \otimes \Phi_2 ) = \max \left\{ \tilde{\tau}_m (\Phi_1), \tilde{\tau}_m (\Phi_2) \right\} .
		\end{equation}
	\end{theorem}
	\begin{proof}
	For two channels $ \Phi_1 $ and $ \Phi_2 $ which output $ n_1 $ and $ n_2 $ modes respectively, let us assume $ \tilde{\tau}_m (\Phi_1) \geqslant \tilde{\tau}_m (\Phi_2) $ and denote $ \tau^\ast = \tilde{\tau}_m (\Phi_1) $. Then, from Proposition \ref{prop:NBCtensor},
		\begin{equation}
		\mathcal{E}_\delta^{\otimes (n_1+n_2)} \circ (\Phi_1 \otimes \Phi_2) = ( \mathcal{E}_\delta^{n_1} \circ \Phi_1 ) \otimes ( \mathcal{E}_\delta^{n_2} \circ \Phi_2 ) 
		\end{equation}
is nonclassicality-breaking if $ \delta \geqslant \tau^\ast $. If $ \delta < \tau^\ast $, $ \mathcal{E}_\delta^{\otimes n_1} $ preserves the nonclassicality of $ \Phi_1 (\rho) $ for the state saturating $ \tau_m \left( \Phi_1(\rho) \right) = \tilde{\tau}_m (\Phi_1) $. Therefore, we have $ \tilde{\tau}_m ( \Phi_1 \otimes \Phi_2 ) = \tau^\ast $.
	\end{proof}
By recalling Observation \ref{obs:preserve}, this theorem implies that using correlated input states does not preserve more nonclassicality depth than using separable input states.

\subsection{\label{sec:outMC}Output maximal coherence of channels}

Now let us consider the maximal coherence discussed in Sec. \ref{sec:maxcoh}. We define the output maximal coherence of channels as
	\begin{equation}
	\tilde{\eta} (\Phi) = \max_\rho \left\{ \eta \left( \Phi (\rho) \right) \right\} .
	\end{equation}
We here make a conjecture on the tensorization property of $ \tilde{\eta} $.
	\begin{conjecture} \label{conj:outMCtensor}
	The output maximal coherence of channels $ \tilde{\tau}_m $ satisfies the tensorization property, that is,
		\begin{equation} \label{eq:outMCtensor}
		\tilde{\eta} ( \Phi_1 \otimes \Phi_2 ) = \max \left\{ \tilde{\eta} (\Phi_1), \tilde{\eta} (\Phi_2) \right\} .
		\end{equation}
	\end{conjecture}
Whether this conjecture holds true is not straightforward to see due to the optimization over all input states. However, we find a strong evidence to support this conjecture by conducting numerical calculations. For this purpose, we have generated pairs of random dynamical matrix (Choi matrix) using the function in QI package for Mathematica \cite{QI}. A pair of dynamical matrices respectively correspond to channels $ \Phi_A $ and $ \Phi_B $ in the spirit of Choi-Jamiolkowski isomorphism \cite{Jamiolkowski1972,Choi1975}. The output maximal coherence is estimated for $ \Phi_A $, $ \Phi_B $, and $ \Phi_A \otimes \Phi_B $ respectively by numerically evaluating Eq. (\ref{eq:outMCtensor}). We plot the obtained data in Fig. \ref{fig:numerical} with the dimension of Hilbert space set as $ \{ d_A, d_B \} = \textrm{(a)}\{2,2\}, \textrm{(b)}\{2,3\} $ ,respectively.
	\begin{figure}[t]
	\centering \includegraphics[clip=true, width=\columnwidth]{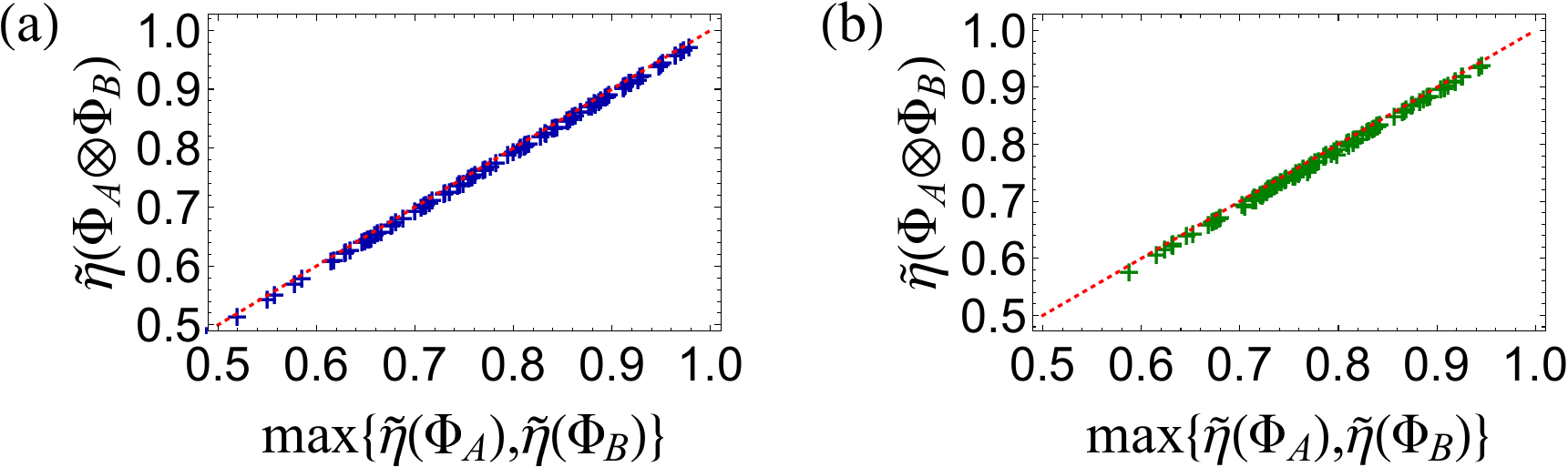}
	\caption{\label{fig:numerical} Numerical calculation of output maximal coherence through channels. The dimension of the Hilbert space of system, $ \{ d_A, d_B \} $, is given by (a)\{2,2\} and (b)\{2,3\}, respectively. Straight red dotted lines represent the tensorization property (\ref{eq:outMCtensor}).}
	\end{figure}
Our data coincides with the diagonal line almost perfectly, which confirms Conjecture \ref{conj:outMCtensor}.
	
Recalling Theorem \ref{thm:cohdistill}, the tensorization property of $ \tilde{\eta} $ would have important consequences in SIO coherence distillation.
	\begin{conjecture} \label{conj:outcohdistill}
	\begin{enumerate}[label=(\alph*)]
	\item If channels $ \tilde{\eta} ( \Phi_j ) < 1 $ for all $ j = 1, 2, \cdots, N $, joint use of such channels in parallel outputs only SIO-nondistillable states.
	\item When a state $ \rho $ is prepared under joint channels $ \Phi_1 \otimes \Phi_2 \times \cdots \otimes \Phi_N $, the maximum distillation fidelity under SIO is given by
		\begin{equation}
		\lim_{n \to \infty} F_\textrm{SIO} ( \rho^{\otimes n}, 2 ) = \frac{1 + \max_j \left\{ \tilde{\eta}(\Phi_j) \right\}}{2} .
		\end{equation}
	\end{enumerate}
	\end{conjecture}
This demonstrates that using joint channels with correlated input states does not improve the performance of SIO coherence distillation of output states compared with using single channel with the maximum $ \tilde{\eta} $.

\section{\label{sec:discussion}Discussion}

% other measures
%% non-Gaussian state manipulation
%% operational resource theory
In this work, we have studied the tensorization property of quantum resources to find the limitations on manipulating quantum resources. If a resource measure satisfies the tensorization proprty as well as the usual monotonicity, it is impossible to concentrate multiple noisy states to a single state with a higher degree of resource by free operations. Furthermore, we have introduced the output resource measure of channels which satisfies the tensorization property to show that joint channels with correlated input states are not helpful to preserve quantum resources. We have established our results in the general framework of QRTs so that it can be applied to any quantum resources. Numerous resources studied intensively so far satisfy the tensorization property, as we have illustrated our results with nonclassicality depth, metrological power, Gaussian quantum resource and maximal coherence. 

It would be an interesting study to further investigate whether there can exist certain resource measures which satisfy the tensorization property in other resource theories such as quantum non-Gaussianity \cite{Takagi2018,Albarelli2018,Park2019} or magic states \cite{Veitch2014,Howard2017}. Another direction of study to pursue is to find operational resource measures satisfying the tensorization property. As the metrological power of nonclassicality has a significant operational meaning in quantum metrology, it may be worth investigating whether one can find a similar measure in the resource theory of coherence.

% connection between tensorization property of states and of channels
We find a close connection between the tensorization property of the state resource measure and that of the channel resource measure for the case of nonclassicality depth and the maximal coherence. It is another interesting question whether there exists a resource measure which satisfies the tensorization property as a state measure but does not as a channel output measure, which can have an analogy to the additivity problems of quantum channels.

% other restrictions in quantum resource manipulation
We hope our work could provide some useful insight into what is allowed or prohibited in manipulating quantum resources under noisy circumstances. Another important contribution has recently been made on the no-go theorem for resource purification \cite{Fang2020}. All of these studies will make a crucial basis for developing practical protocols to overcome unavoidable noise in manipulating resources for quantum informational tasks.

\section*{Acknowledgement}
J.L., K.B., and J.K. are supported by KIAS Individual Grants (No. CG073102, No. CG074701, and No. CG014604) at Korea Institute for Advanced Study, respectively. J.P. and K.B. acknowledge supports by the National Research Foundation of Korea (NRF) grants (No. NRF-2019R1G1A1002337) and (Nos. NRF-2021M3E4A1038213, NRF-2022M3E4A1077094) funded by the Korean government (MSIT), respectively. K.B. is supported by the Ministry of Science, ICT and Future Planning (MSIP) by the Institute of Information and Communications Technology Planning and Evaluation grant funded by the Korean government (No. 2019-0-00003, “Research and Development of Core Technologies for Programming, Running, Implementing and Validating of Fault-Tolerant Quantum Computing System”) and (No. 2022-0-00463, “Development of a quantum repeater in optical fiber networks for quantum internet”). H.N. is supported by an NPRP Grant 13S-0205-200258 from Qatar National Research Fund.

\appendix

\section{\label{sec:NDproperty}Properties of nonclassicality depth}

\emph{Faithfulness:} Faithfulness (\textit{C}1b) as well as (\textit{C}1a) holds by the definition of the nonclassicality depth.

\emph{Tensorization property:} Let us assume, without loss of generality, $ \tau_m(\rho) \geqslant \tau_m(\sigma) $ and $ \tau^\ast \equiv \tau_m(\rho) $. Observing that
	\begin{equation} \label{eq:Wproduct}
	W_{\rho \otimes \sigma} (\bm{\alpha}; \tau) = W_\rho (\bm{\alpha}_1; \tau) W_\sigma (\bm{\alpha}_2; \tau) ,
	\end{equation}
where $ \bm{\alpha} $ denotes the collection of $ \bm{\alpha}_1 $ and $ \bm{\alpha}_2 $, becomes a positive probability function if $ \tau = \tau^\ast $, we find $ \tau_m (\rho \otimes \sigma) \leqslant \tau^\ast $. Further, if $ \tau < \tau^\ast $, $ W_\rho (\bm{\alpha}_1; \tau) $ must show a negative value, and thus Eq. (\ref{eq:Wproduct}) cannot be a positive function. Therefore, we have the tensoriztaion property (\textit{C}4) written as
	\begin{equation} \label{eq:NDtensorization}
	\tau_m (\rho \otimes \sigma) = \max \{ \tau_m(\rho), \tau_m(\sigma) \} .
	\end{equation}

\emph{Monotonicity:} We prove the strictest form of  monotonicity (\textit{C}2c) and then (\textit{C}2a) and (\textit{C}2b) automatically follow. A passive linear unitary operation $ \hat{U}_P $ can be regarded as a rotation in phase space, which is described by the transformation $ W_{\hat{U}_P\rho\hat{U}_P^\dagger} (\bm{R}\cdot\bm{\alpha}; \tau) = W_\rho (\bm{\alpha}; \tau) $, where $ \bm{R} $ is an orthogonal matrix satifying $ \bm{R}^\mathsf{T} = \bm{R}^{-1} $ and $ \det\bm{R} = 1$. Because the rotation preserves positivity of the function $ W $, the nonclassicality depth $ \tau_m $ is invariant under passive linear unitaries. Similarly, $ \tau_m $ is invariant under displacements $ \hat{D}(\bm{\gamma}) $ described by the transformation $ W_{\hat{D}(\bm{\gamma})\rho\hat{D}^\dagger(\bm{\gamma})} (\bm{\alpha}; \tau) = W_\rho (\bm{\alpha}-\bm{\gamma}; \tau) $. The nonclassicality depth is also invariant under addition of classical ancilla modes due to the tensorization property (\ref{eq:NDtensorization}). Let us now consider a classical measurement performed on the last $ k $ modes out of an $n$-mode state $ \rho $. Projection onto multimode coherent states is described by $ \hat{M}_{\bm{\xi}} \equiv \frac{1}{\pi^k} \ket{\bm{\xi}}\bra{\bm{\xi}} $, where $ \ket{\bm{\xi}} $ represents $k$-mode coherent states. After measurement, the state can be written by the following P representation:
	\begin{eqnarray}
	\rho' & = & \frac{1}{\pi^k p(\bm{\xi}|\rho)} \bra{\bm{\xi}} \rho \ket{\bm{\xi}}  \nonumber \\
	& = & \frac{1}{\pi^k p(\bm{\xi}|\rho)} \int d^{2n}\bm{\alpha} \braket{\bm{\xi}}{\bm{\alpha}} P_\rho (\bm{\alpha}) \braket{\bm{\alpha}}{\bm{\xi}} \nonumber \\
	& \equiv & \int d^{2(n-k)}\bm{\alpha}\prime P_{\rho'} (\bm{\alpha}') \ket{\bm{\alpha}'}\bra{\bm{\alpha}'} ,
	\end{eqnarray}
where $ p(\bm{\xi}|\rho) = \tr \bra{\bm{\xi}} \rho \ket{\bm{\xi}} $ and
	\begin{widetext}
	\begin{equation}
	P_{\rho'} (\bm{\alpha}') = \frac{1}{\pi^k p(\bm{\xi}|\rho)} \int d^{2k}\overline{\bm{\alpha}} e^{-| \overline{\bm{\alpha}}-\bm{\xi} |^2} P_\rho (\alpha_1, \alpha_2 \cdots, \alpha_{n-k}, \overline{\alpha}_1, \overline{\alpha}_2, \cdots, \overline{\alpha}_k) ,	
	\end{equation}
with $ \overline{\bm{\alpha}} = \{ \overline{\alpha}_1, \overline{\alpha}_2, \cdots, \overline{\alpha}_k \} $. Then, $s$-parametrized quasiprobability distribution of $ \rho' $ becomes
	\begin{eqnarray}
	W_{\rho'} (\bm{\alpha}'; \tau) & = & \left( \frac{1}{\pi\tau} \right)^{n-k} \int d^{2(n-k)}\bm{\beta} P_{\rho'} (\bm{\beta}) e^{-\frac{1}{\tau} | \bm{\beta}-\bm{\alpha}' |^2 } \nonumber \\
	& = & \frac{1}{\pi^n \tau^{n-k} p(\bm{\xi}|\rho)} \int d^{2(n-k)}\bm{\beta} e^{-\frac{1}{\tau} | \bm{\beta}-\bm{\alpha}' |^2 } \int d^{2k}\overline{\bm{\alpha}} e^{-| \overline{\bm{\alpha}}-\bm{\xi} |^2}  P_\rho (\beta_1, \beta_2 \cdots, \beta_{n-k}, \overline{\alpha}_1, \overline{\alpha}_2, \cdots, \overline{\alpha}_k) \nonumber \\
	& = & \frac{\tau^k}{p(\bm{\xi}|\rho)} \int d^{2k}{\bm{\xi}}' e^{-\frac{1}{1-\tau} | \bm{\xi}'-\bm{\xi} |^2} W_{\rho} (\alpha'_1, \alpha'_2, \cdots, \alpha'_{n-k}, \xi'_1, \xi'_2, \cdots, \xi'_k; \tau) .
	\end{eqnarray}
	\end{widetext}
If $ \tau = \tau_m(\rho) $, $ W_{\rho} (\bm{\alpha}; \tau) $ is positive so that $ W_{\rho'} (\bm{\alpha}'; \tau) $ is positive as well. Therefore, we have $ \tau_m (\rho') \leqslant \tau_m(\rho) $. In the classical mixing process $ \sigma = \sum_i \Phi_i (\rho) $, we have
	\begin{equation}
	W_\sigma (\bm{\alpha}; \tau) = \sum_i W_{\Phi_i(\rho)} (\bm{\alpha}; \tau) .
	\end{equation}
Although $ W_{\Phi_i(\rho)} $ on the right-hand side are not normalized, positivity is not affected by the normalization factor. Let $ \tau^\ast \equiv \max \left\{ \tau_m (\Phi_i(\rho)) \right\} $, then $ W_\sigma (\bm{\alpha}; \tau^\ast) $ is positive, which implies that $ \tau_m (\sigma) \leqslant \tau^\ast \leqslant \tau_m (\rho) $. We have shown that the nonclassicality depth is nonincreasing under all classical processes we considered as free operations. Especially, the nonclassicality depth is nonincreasing under postselection of measurement outcome, which implies the monotonicity (\textit{C}2c).

\section{\label{sec:NBCproof}Proof of Lemma \ref{lem:NBCexp}.}

We first prove that NBCs are necessarily entanglement-breaking. Since a quantum channel $ \Phi $ is linear, it suffices to prove the lemma for pure states only. Suppose a $(n+m)$-mode pure entangled state written in the Schmidt decomposition as $ \ket{\Psi} = \sum_i c_i \ket{u_i} \ket{v_i} $. Let us consider a pure $n$-mode state $ \ket{\phi} = \sum_i d_i \ket{u_i} $, which is a superposition of Schmidt basis $ \ket{u_i} $. The application of NBC $ \Phi_\textrm{NB} $ on $ \ket{\phi} $ reads
	\begin{equation}
	\Phi_\textrm{NB} (\ket{\phi}\bra{\phi}) = \frac{1}{\pi^n} \int d^{2n}\bm{\alpha} \sum_{i,j} d_i P_{ij} (\bm{\alpha}) d_j^\ast \ket{\bm{\alpha}}\bra{\bm{\alpha}} ,
	\end{equation}
where $ P_{ij} (\bm{\alpha}) $ is a P function of $ \Phi_\textrm{NB} (\ket{u_i}\bra{u_j}) $. Because the channel is nonclassicality-breaking, the term $ \sum_{i,j} d_i P_{ij} (\bm{\alpha}) d_j^\ast $ should be nonnegative for any state $ \ket{\phi} $, which is fulfilled if and only if the matrix with elements $ P_{ij}(\bm{\alpha}) $ is positive semidefinite. The application of $ \Phi_\textrm{NB} $ on the first $n$ modes of $ \ket{\Psi} $ reads
	\begin{eqnarray}
	& & (\Phi_\textrm{NB} \otimes I) (\ket{\Psi}\bra{\Psi}) \nonumber \\
	& & = \frac{1}{\pi^n} \int d^{2n}\bm{\alpha} \sum_{i,j} c_i P_{ij} (\bm{\alpha}) c_j^\ast \ket{\bm{\alpha}}\bra{\bm{\alpha}} \otimes \ket{v_i}\bra{v_j} \nonumber \\
	& & = \frac{1}{\pi^n} \int d^{2n}\bm{\alpha} \ket{\bm{\alpha}}\bra{\bm{\alpha}} \otimes \sum_{i,j} c_i P_{ij}(\bm{\alpha}) c_j^\ast \ket{v_i}\bra{v_j} .
	\end{eqnarray}
As $ \left\{ P_{ij}(\bm{\alpha}) \right\}_{ij} $ is a positive semidefinite matrix, the last term $ \sum_{i,j} c_i P_{ij}(\bm{\alpha}) c_j^\ast \ket{v_i}\bra{v_j} $ represents an unnormalized quantum state. Therefore the application of $ \Phi $ always produces separable output states, that is, $ \Phi $ is entanglement-breaking.

An entanglement-breaking channel can always be expressed in the Holevo form $ \Phi_\textrm{EB} (\rho) = \sum_k \tr \left[ \rho M_k \right] \sigma_k $ \cite{Horodecki2003}, where $ \{ M_k \}_k $ is a set of POVM and $ \sigma_k $'s are density operators. Because $ \Phi_\textrm{NB} $ is also entanglement-breaking, we can write
	\begin{eqnarray}
	\Phi_\textrm{NB} (\rho) & = & \sum_k \tr \left[ \rho M_k \right] \frac{1}{\pi^n} \int d^{2n}\bm{\alpha} P_{\sigma_k}(\bm{\alpha}) \ket{\bm{\alpha}}\bra{\bm{\alpha}} \nonumber \\
	& = & \frac{1}{\pi^n} \int d^{2n}\bm{\alpha} \ket{\bm{\alpha}}\bra{\bm{\alpha}} \tr \left[ \rho \sum_k P_{\sigma_k}(\bm{\alpha}) M_k \right] \nonumber \\
	& \equiv & \int d^{2n}\bm{\alpha} \ket{\bm{\alpha}}\bra{\bm{\alpha}} \tr \left[ \rho M_{\bm{\alpha}} \right] .
	\end{eqnarray}
Because of the definition of NBC, $ \tr \left[ \rho \sum_k P_{\sigma_k}(\bm{\alpha}) M_k \right] $ should be nonnegative for any quantum state $ \rho $. This guarantees that $ M_{\bm{\alpha}} \equiv \frac{1}{\pi^n} \sum_k P_{\sigma_k}(\bm{\alpha}) M_k $ is a positive operator and thus $ \left\{ M_{\bm{\alpha}} \right\}_{\bm{\alpha}} $ is a set of POVM.


\begin{thebibliography}{1}
\bibitem{footnote2} In \cite{Yadin2018}, the metrological power is defined in a more general way, that is, $ F_k = \sum_{i=1}^{k} \max \left\{ \lambda_k (\bm{F}) - \frac{1}{2}, 0 \right\} ~ ( k = 1, 2, \cdots, 2N ) $, where $ \lambda_k (\bm{F}) $'s are eigenvalues of $ \bm{F} $ in decreasing order. In this work, we are only interested in the case of $ k = 1 $.
\end{thebibliography}
\end{document}